\documentclass[a4paper,11pt]{article}

\usepackage{amsfonts}
\usepackage{indentfirst}
\usepackage{a4wide}
\usepackage{amssymb}
\usepackage{graphicx}
\usepackage{graphics}
\usepackage{appendix}
\usepackage{authblk}
\usepackage[utf8]{inputenc}
\usepackage[T1]{fontenc}
\usepackage{color}
\usepackage{indentfirst, amsfonts, amsmath, amsthm, amssymb, amscd}
\usepackage[all]{xy}
\usepackage{braket}
\usepackage{epsfig}
\usepackage[pagebackref=true, bookmarksopen=true,colorlinks=true, linkcolor=red,citecolor=blue]{hyperref}
\usepackage{lineno}

\newtheorem{teo}{Theorem}[section]

\newtheorem{defn}{Definition}[section]
\newtheorem{example}{Example}[section]
\newtheorem{remark}{Remark}

\def\fh{\mathfrak{h}}

\def\cH{\mathcal{H}}

\title{Nilmanifolds and their associated non local fields}
\author{Juan J. Villarreal}

\begin{document}

\maketitle

\begin{abstract}

For a three dimensional nilmanifold together with a three form on it, we build a module $\cH$ of an affine Kac Moody vertex algebras. Then, we associate logarithmic fields to the module $\cH$ and we study their singularities. We also present a physics motivation behind this construction. 

We study a particular case, we show that when the nilmanifold $N$ is a $k$ degree $S^1$--fibration over the two torus and a choice of $l \in \mathbb{Z} \simeq H^3(N, \mathbb{Z})$ the fields associated to the space $\cH$ have tri-logarithm singularities whenever $kl \neq 0$. 

\end{abstract}


\section{Introduction}

Vertex algebras appeared in the early days of string theory, in string theory vertex algebras can be seen as a mathematical counterpart of chiral symmetry algebras in conformal field theory (CFT). Working the CFT interpretation in physics of sigma models, we can associate to some manifolds a vertex algebra. This interpretation turns out to be very restrictive, in particular, the manifolds must be flat\footnote{If we work with super vertex algebras then the manifolds could be Calabi Yau manifolds, these restrictions are given by the beta equations in the physics literature,  \cite{green1988superstring}.}. We can consider other interpretations to associate vertex algebras to manifolds. In \cite{alekseev2005current} the authors considered a bracket for some \emph{fields} that naturally leads to the Courant bracket or more generally, considering additionally a closed three form, the twisted Courant bracket. In this interpretation the bracket does not contain dynamical information. In this work, we use vertex algebras to describe this bracket, we call this construction of vertex algebras the Hamiltonian formalism\footnote{If we consider super vertex algebras, there is a construction called Chiral de Rham \cite{malikov1999chiral} which associate sheaves of super vertex algebras to any manifold. The Courant bracket also appear in this construction, there is an Hamiltonian interpretation of this construction in \cite{ekstrand2009non} }.



In \cite{aldi2012dilogarithms} the authors studied from this point of view the algebras associated to the three dimensional Heisenberg nilmanifolds and , dually,  to the three dimensional torus with a closed three form. This informations is used to build some 6 dimensional nilmanifolds, $M(1,0)$ and $M(0,1)$ respectively wich will be defined in section \ref{3333}. The motivation behind the choice of these manifolds and three forms is a phenomenon in physics called T-duality, the 6 dimensional nilmanifolds also have a motivation from a physical theory called Double field theory \cite{hull2009double}. 

In this work we consider a more general case, three dimensional Heisenberg nilmanifolds with closed three forms. In this case, we have more general 6 dimensional nilmanifolds $M(k,j)$. We explain the construction of these nilmanifolds in the section \ref{3333}, the construction of these nilmanifolds from a physics point of view was given in \cite{reid2009flux}. Considering the Hamiltonian formalism for these more general manifolds, we generalize the algebras found in \cite{aldi2012dilogarithms} as we explain below. 

Our interest behind this construction is to understand, in the framework of vertex algebras, the algebraic structure associated to some \emph{logarithmic fields} on these nilmanifolds. These logarithmic fields describe infinite dimensional Lie algebras that we express in terms of singularities, in particular these fields are not local.    

Before we explain the algebraic construction associated to these six dimensional nilmanifolds, we set our notation for vertex algebras.
 
    

In vertex algebras theory, some of the first examples we study are Heisenberg vertex algebra and affine Kac Moody algebras. We can define these vertex algebras from finite dimensional Lie algebras $\mathfrak{h}$ endowed with a bi-invariant pairing. The space of states of these vertex algebras is given by the vector space\footnote{We used the notation in \cite{frenkel2004vertex}, where the loop algebra is defined as $L\mathfrak{h}:=\mathfrak{h}\otimes \mathbb{C}((t))$ and the algebra $\hat{\mathfrak{h}}\approx L\mathfrak{h}\oplus K\mathbb{C}$ as the central extension for the bi-invariant pairing.}.
$$V^{1}(\mathfrak{h})=Ind^{\hat{\mathfrak{h}}}_{\mathfrak{h}[[t]]\oplus \mathbb{C}K}\ket{0}=U(\hat{\mathfrak{h}})\otimes_{U(\mathfrak{h}[[t]]\oplus\mathbb{C}K)}\mathbb{C}\ket{0},$$
where $K$ acts as $K=1Id$. The fields are defined as linear maps
\begin{equation}Y(.,z):V^{1}(\mathfrak{h})\rightarrow End(V^{1}(\mathfrak{h}))[[z,z^{-1}]],\quad s.t.\quad Y(a,z)b\in V^{1}(\mathfrak{h})((z)).\label{0.01}\end{equation}

In this work we consider six dimensional Lie algebras $\mathfrak{h}$ with basis $\{\alpha_{i},\beta_{j}\}$ for $i,j\in \{1,2,3\}$ and bi-invariant pairing $(\alpha_{i},\beta_{j})=\delta_{ij}$. The generating fields are defined as
\begin{equation}\alpha_{i}(z)=\sum_{n\in\mathbb{Z}}\alpha^{i}_{n}z^{-n-1},\quad\quad\beta_{i}(z)=\sum_{n\in\mathbb{Z}}\beta^{i}_{n}z^{-n-1},\label{generating}\end{equation}
where we consider the basis ${\alpha^{i}_{n}=\alpha_{i}\otimes t^{n}}$ and ${\beta^{i}_{n}=\beta_{i}\otimes t^{n}}$ for $i=1,2,3$ and $n\in\mathbb{Z}$.

\subsection{Torus case}
Consider a six dimensional torus, in our notation $\mathbb{T}^{6}=M(0,0)$. Now $\mathbb{T}^{6}=\mathfrak{H}/\Lambda$ where $\mathfrak{H}$ is $\mathbb{R}^{6}$ endowed with the abelian group structure of the sum, $\Lambda$ is a discrete subgroup. We consider the abelian Lie algebra $Lie(\mathfrak{H})=\mathfrak{h}$ and we construct out of this the Heisenberg vertex algebra $V^{1}(\mathfrak{h})$, we express the algebraic relations between its generating fields (\ref{generating}) as follows
$$[\alpha_{i}(z),\beta_{j}(w)]=\delta_{ij}\partial _{w}\delta (z-w)\quad\quad i, j \in \{1,2,3\}.$$

Considering the action of the Lie algebra $\mathfrak{h}$ on the function space of the torus $\mathbb{T}^{6}$, we define the space
$$\mathcal{H} = \mathrm{Ind}^{\hat{\mathfrak{h}}}_{\mathfrak{h}[[t]]\oplus \mathbb{C}K}L^{2}(\mathbb{T}^{6}) \simeq U(\hat{\fh}) \otimes_{U(\mathfrak{h}[[t]]\oplus \mathbb{C}K)} L^{2}(\mathbb{T}^{6}).$$
By definition $\mathcal{H}$ is a module for the vertex algebra $V^{1}(\mathfrak{h})$. In some cases $\mathcal{H}$ has the structure of vertex algebra the \emph{lattice vertex algebra}.
We are motivated from the physical interpretation, to consider fields $x_{i}(z)$ and $y_{j}(z)$ associated to coordinates $x_{i}$ and $y_{i}$ on the torus $\mathbb{T}^{6}$, these fields satisfy 
\begin{equation}\partial_{z}x_{i}(z)=\alpha_{i}(z),\quad\quad \partial_{z}y_{i}(z)=\beta_{i}(z).\label{no1}\end{equation}

Then we consider the \emph{logarithmic fields} in $End(\mathcal{H})[[z,z^{-1}]][\log z]$ 
\begin{equation}x_{i}(z)=w_{i}\log z+\sum_{n\in\mathbb{Z}}x^{i}_{n}z^{-n},\quad y_{i}(z)=p_{i}\log z+\sum_{n\in\mathbb{Z}}y^{i}_{n}z^{-n},\label{no2}\end{equation}
subject to the algebraic relation 
\begin{equation}[\partial_{z}x_{i}(z), \partial _{w} y_{j}(w)]=[\alpha_{i}(z),\beta_{j}(w)]=\delta_{ij}\partial _{w}\delta (z-w),\label{2}\end{equation}
therefore
\begin{equation}[x_{i}(z),y_{j}(w)]=\delta_{ij} {\log (z-w)} \label{2correction}\end{equation} 
where the singularity is defined by    
{\begin{align}
\begin{split}
{\log (z-w)}&:=i_{z,w}\log (z-w)-i_{w,z}\log (w-z)=\log z+\log \left( 1-\frac{w}{z}\right)-\log w-\log \left( 1-\frac{z}{w}\right)\\
&{=\log z-\sum_{n>0}\frac{w^{n}z^{-n}}{n}-\log w+\sum_{n>0}\frac{z^{n}w^{-n}}{n}},
\end{split}
\label{log}
\end{align}}the notations $i_{z,w}$ and $i_{w,z}$ denote the expansion for $z>w$ and $w>z$ respectively. This is known as a \emph{logarithmic  singularity}. It is easy to see from (\ref{2}) that the modes $\{w_{i},p_{i},x^{i}_{n},y^{i}_{n}\}_{n\in\mathbb{Z}}$ for $i\in\{1,2,3\}$ satisfy a Lie algebra.

\subsection{Twisted torus case}

In this case, we consider the six dimensional nilmanifold\footnote{The case $M(1,0)$ is similar to the twisted torus case, the reason behind this is the T-duality} $M(0,1)$, this nilmanifold is called the double twisted torus. Now $M(0,1)=\mathfrak{H}/\Lambda$ where $\mathfrak{H}$ is $\mathbb{R}^{6}$ endowed with a two step nilpotent structure, and $\Lambda$ is a discrete subgroup. The nilpotent Lie algebra $Lie(\mathfrak{H})=\mathfrak{h}_{0,1}$ is given by
$$[\beta_{i},\beta_{j}]=\epsilon_{ijk}\alpha_{k},\quad[\alpha_{i},\beta_{j}]=0,\quad[\alpha_{i},\alpha_{j}]=0\quad\quad i, j\in \{1,2,3\}.$$
Where $\epsilon_{ijk}$ is the antisymetric tensor. We associate to this Lie algebra an affine Kac Moody vertex algebra $V^{1}(\mathfrak{h}_{0,1})$, the algebraic relations between its generating fields (\ref{generating}) are given by
\begin{equation}[\beta_{i}(z),\beta_{j}(w)]=\epsilon_{ijk}\delta(z-w)\quad\quad[\alpha_{i}(z),\beta_{j}(w)]=\delta_{ij}\partial _{w}\delta (z-w)\quad\quad i, j \in \{1,2,3\}.\label{no4}\end{equation}

The nilpotent algebra $\mathfrak{h}_{0,1}$ acts on the nilmanifold $M(0,1)$ and analogously we define a module for our vertex algebra 
$$\mathcal{H}_{0,1} = \mathrm{Ind}^{\hat{\mathfrak{h}}}_{\mathfrak{h}_{0,1}[[t]]\oplus \mathbb{C}K}L^{2}(M(0,1)) \simeq U(\hat{\fh}) \otimes_{U(\mathfrak{h}_{0,1}[[t]]\oplus \mathbb{C}K)} L^{2}(M(0,1)).$$ 

We are motivated from the physical interpretation, to consider fields $x_{i}(z)$ and $y_{j}(z)$ associated to coordinates $x_{i}$ and $y_{i}$ on the nilmanifold $M(0,1)$. The nilmanifold $M(0,1)$ comes with a global framing on its cotangent bundle given by $T^{*}M(0,1)\approx \mathfrak{h}_{0,1}\times M(0,1)$, then we express with coordinates $\{x_{i},y_{i}\}$ on $M(0,1)$ a basis $\{\alpha_{i},\beta_{j}\}$ of $T^{*}M(0,1)$.
$$\alpha_{i}=dx_{i},\quad \beta_{i}=dy_{i}-\frac{1}{2}\varepsilon_{ijk}x_{j}dx_{k}.$$
Then, also, these are the relations between their fields
\begin{equation} \partial _{z}x_{i}(z)=\alpha_{i}(z), \quad\partial_{z}y_{i}(z)=\beta_{i}(z)+\frac{1}{2}\varepsilon_{ijk}:x_{j}(z)\partial_{z}x_{k}(z){ :}.\label{no3}\end{equation}
We can find the logarithmic fields $End(\mathcal{H}_{0,1})[[z,z^{-1}]][\log z]$ associated to the coordinates integrating the relations before\footnote{These fields have an interpretation related with the loop space $LM(0,1)$ following the Hamiltonian interpretation},
$$x_{i}(z)=w_{i}\log z+\sum_{n\in\mathbb{Z}} x^{i}_{n}z^{-n},\quad\quad y_{i}(z)=p_{i}\log z+\sum_{n\in\mathbb{Z}} y^{i}_{n}z^{-n}+\frac{\varepsilon_{ijk}}{2}w_{j}x_{k}(z)\log z.
$$
The algebraic relations between these fields are restricted by their relations with the fields $\{\alpha_{i}(z),\beta_{j}(z)\}$ in (\ref{no3}) and (\ref{no4}) in the same way that happens in the torus case (\ref{2}). We emphasize here that the logarithmic singularity (\ref{log}) is not enough to express the algebraic relations for these fields, for example the relation between the fields $y_{i}(z)$ and $y_{j}(w)$ satisfy a relation that has the following form 
\begin{align*}
    &[y_{i}(z),y_{j}(w)]=(...)\log(z-w)+\frac{1}{2}\epsilon_{ijk}w_{k}rl(z,w).
\end{align*}
Here the notation $(...)$ means some expressions that involve fields, we express the complete relations in (\ref{rll}). In this case the Roger's dilogarithm $rl(x):=Li_{2}(x)+\frac{1}{2}\log x\log(1-x)$ appears naturally as a singualrity
\begin{align}
\begin{split}
{ {rl(z,w)}}=&Li_{2}(\frac{z}{w})+\frac{1}{2}\log(\frac{z}{w})\log(1-\frac{z}{w})+Li_{2}(\frac{w}{z})+\frac{1}{2}\log(\frac{w}{z})\log(1-\frac{w}{z})\\
&=Li_{2}(\frac{z}{w})+Li_{2}(\frac{w}{z})+\frac{1}{2}(\log w-\log z)^{2}-\frac{1}{2}(\log z-\log w)\log (z-w).
\end{split}
\label{5}
\end{align}

In this case also the modes $\{w_{i},p_{i},x^{i}_{n},y^{i}_{n}\}_{n\in\mathbb{Z}}$ for $i\in\{1,2,3\}$ form a Lie algebra, we can see this Lie algebra as a particular case of the algebra that we describe in the section \ref{sec:3.1}. 

The twisted torus case was studied in \cite{aldi2012dilogarithms} with the Hamiltonian formalism, and from a perturvative CFT point of view in \cite{blumenhagen2011non}. In these works, the authors also work around the interpretation of this singularity and its identities to explain properties in vertex algebras and CFT.  We show in the appendix the relation between these two formalism working their physical interpretations, on the one hand we have the Hamiltonian point of view and on the other hand we have an CFT perturvative point of view. 

\subsection{Twisted nilmanifolds case}

In this work, we consider the more general case given by the six dimensional nilmanifolds $M(k,j)$. The natural appearance of the Roger's dilogarithm on the previous case motivated us to study the singularities of these logarithmic fields in more general cases. Also in this case we do not have perturvative CFT interpretation therefore we have only the Hamiltonian formalism.  

We consider the six dimensional nilmanifold $M(k,j)$, also known as the twisted nilmanifold case. Now $M(k,j)=\mathfrak{H}/\Lambda$ where $\mathfrak{H}$ is $\mathbb{R}^{6}$ endowed with a three step nilpotent structure, and $\Lambda$ is a discrete subgroup. The nilpotent Lie algebra $Lie(\mathfrak{H})=\mathfrak{h}_{k,j}$ is given by
\begin{align}
\begin{split}  [\beta_{1},\beta_{2}]=j\alpha_{3},\quad [\beta_{3},\beta_{1}]=j\alpha_{2}+k\beta_{2},\quad[\beta_{2},\beta_{3}]=j\alpha_{1},\\
 [\beta_{1},\alpha_{2}]=k\alpha_{3},\quad [\beta_{2},\alpha_{2}]=0,\quad [\alpha_{2},\beta_{3}]=k\alpha_{1}.
\end{split}
\label{0.0}
\end{align}
 We associate to this Lie algebra an affine Kac Moody vertex algebra $V^{1}(\mathfrak{h}_{0,1})$, we express the algebraic relations between its generating fields (\ref{generating}) as follows
 \begin{align}
\begin{split}
   [\beta_{1}(z),\beta_{2}(w)]= j\alpha_{3}(w)\delta(z-w),\quad[\beta_{2}(z),\beta_{3}(w)]= j\alpha_{1}(w)\delta(z-w),\\ 
  [\beta_{1}(z),\alpha_{2}(w)]=k\alpha_{3}(w)\delta(z-w),\quad[\alpha_{2}(z),\beta_{3}(w)]=k\alpha_{1}(w)\delta(z-w),\\
   [\beta_{3}(z),\beta_{1}(w)]= \left(j\alpha_{2}(w)+k\beta_{2}(w)\right)\delta(z-w),\quad  [\alpha_{i}(z),\beta_{i}(w)]=\partial_{w}\delta(z-w).
   \end{split}
\label{1}
\end{align}
The nilpotent algebra $\mathfrak{h}_{k,j}$ acts on the nilmanifold $M(k,j)$ and analogously we define a module for our vertex algebra  
$$\mathcal{H}_{k,j} = \mathrm{Ind}^{\hat{\mathfrak{h}}}_{\mathfrak{h}_{k,j}[[t]]\oplus \mathbb{C}K}L^{2}(M(k,j) \simeq U(\hat{\fh}) \otimes_{U(\mathfrak{h}_{k,j}[[t]]\oplus \mathbb{C}K)}L^{2}(M(k,j)).$$ 

Following the same idea in (\ref{no3}) we express the fields $\{x_{i}(z),y_{j}(z)\}$ in term of the fields  $\{\alpha_{i}(z),\beta_{j}(z)\}$, see (\ref{modos}). We find a new singularity for the fields $\{x_{i}(z),y_{j}(z)\}$ which satisfy the restriction impose by their relations with the algebra \eqref{1} of the fields $\{\alpha_{i}(z),\beta_{j}(z)\}$. For example the relation between the fields $y_{1}(z)$ and $y_{3}(w)$ satisfy a relation that has the following form
$$[y_{1}(z),y_{3}(w)] = (...)\log (z-w) +(...)rl(z,w)-kjw_{3}w_{1}t(z,w),$$
here the notation (...) means some expressions that involve fields. We express the complete relations in the theorem \ref{principal}. In this case the new singularity is given by 
\begin{align}
\begin{split}{ {t(z,w)}}=&-2\left( Li_{3}(\frac{z}{w})-Li_{3}(\frac{w}{z})\right)+(\log z-\log w)\left( Li_{2}(\frac{z}{w})+Li_{2}(\frac{w}{z})\right)\\
&+\frac{1}{6}(\log z-\log w)^{3}-\frac{1}{6}(\log z^{2}-3\log w\log z+\log w^{2})\log(z-w).\end{split}
\label{5.5}
\end{align}
In this case the modes $\{w_{i},p_{i},x^{i}_{n},y^{i}_{n}\}_{n\in\mathbb{Z}}$ for $i\in\{1,2,3\}$ form a non-linear Lie algebra, we give the definition of non-linear Lie algebras in the section \ref{Non linear} and we express the non-linear Lie algebra that we found in theorem \ref{theorem4.2}.

Now we emphasize the main new results in this work. First, we found a new example of logarithmic fields describing non local algebraic relations, theorem \ref{principal}. Second, the algebra satisfied by these logarithmic fields shows a new kind of polylogarithm singularity \eqref{5.5}, in particular the tri-logarithm function $Li_{3}(x)$ appears. This singularity is new compared to \cite{aldi2012dilogarithms} where only up to dilogarithmic functions appeared. Third, the Lie algebra of modes defined by the logarithmic fields is a non linear Lie algebra, theorem \ref{theorem4.2}.

This work is organized as follows. In section \ref{2222}, we give short physics motivation. In section \ref{3333}, we introduce and define the objects that we use. In section \ref{chap3}, we express the fields and the algebra that we have from studying the general case $k\neq 0 $ and $j\neq 0$.\\
\\

\textbf{Acknowledgment}: I thank Reimundo Heluani for his continuous support while this work was done, also I want to thank thank Jethro Van Ekeren, Pedram Hekmathi and Alejandro Cabrera for inspiring conversations. This work was done at IMPA, and was finally completed at Virginia Commonwealth University with the help of Marco Aldi. I would like to thank
these institutions for hospitality and excellent working conditions.

\section{ Physics motivation}\label{2222}

In this section we give a very short motivation to some constructions done in this work. A more complete treatment of the concepts introduced here is in the cited references.

\subsection{The current algebras}
In the CFT formalism on the torus there is a symmetry called T-duality, for a vertex algebraic approach see \cite{kapustin2003vertex}. On the other hand, an analogous study of symmetries as T-duality becomes more complicated in general cases. In particular, we are interested in cases where instead of the torus we have nilmanifolds. There is an interesting relation at the topological level of T-duality between three dimensional Heisenberg nilmanifolds $N(k)$ with $H_{j}$-\emph{flux}, \cite{bouwknegt2004t}. The nilmanifolds  $N(k)$ is a $S^{1}$-bundle   
$$S^{1}\rightarrow N(k)\rightarrow \mathbb{T}^{2},$$
where $k\in\mathbb{Z}\simeq H^{2}(\mathbb{T}^{2},\mathbb{Z})$ is the Chern class. And the H-\emph{flux} is given by a three form $H_{j}$ s.t $[H_{j}]= j\in\mathbb{Z}\simeq H^{3}(N(k),\mathbb{Z})$. For these nilmanifolds we can try to develop a theory similar to the case of the torus.  However, several problems arise. In particular, only on the torus we have a CFT. Therefore, we consider a different approach, the Hamiltonian formalism.

From this point of view, we have a relation between Poisson brackets of certain fields called \emph{currents} and Courant brackets \cite{alekseev2005current}.  We now describe this relationship in vertex algebras for the nilmanifolds.\footnote{The relation between between Courant brackets and T-duality, for example in \cite{cavalcanti2011generalized}, can be associated to the Poisson brackets in \cite{alekseev2005current}, this approach was considered in \cite{hekmati2012t}.} We will restrict ourselves to consider only the global sections of $TN(k)\oplus T^{*}N(k)$ given by
\begin{align*}
  {\beta_{1}}\color{black}=\partial_{x}+\frac{k}{2}z\partial_{y},\,\,\,  {\beta_{2}}\color{black}=\partial_{y},\,\,\,  {\beta_{3}}\color{black}=\partial_{z}-\frac{k}{2}x\partial_{y}, \\ 
  {\alpha_{1}}\color{black}=dx,\,\,\,  {\alpha_{2}}\color{black}=dy+\frac{k}{2}xdz-\frac{k}{2}zdx,\,\,\,  {\alpha_{3}}\color{black}=dz.
\end{align*}
The Courant brackets between these sections are given by the Lie algebra (\ref{0.0}). We denoted this Lie algebra as $\mathfrak{h}_{k,j}$, this Lie algebra comes with the bi-invariant pairing given by $(\alpha_{i},\beta_{j})=\delta_{ij}$. The current algebra in \cite{alekseev2005current} can be described in the vertex algebra formalism as the following Kac Moody vertex algebra 
$$V^{1}(\mathfrak{h}_{k,j})=Ind^{\hat{\mathfrak{h}}_{k,j}}_{\mathfrak{h}_{k,j}[[t]]\oplus \mathbb{C}K}\ket{0}=U(\hat{\mathfrak{h}}_{k,j})\otimes_{U(\mathfrak{h}_{k,j}[[t]]\oplus\mathbb{C}K)}\mathbb{C}\ket{0}$$
where $K$ acts as $K=1Id$.  The generating fields satisfy the algebra given in (\ref{1}).\footnote{Here the variable z is the standard notation in the vertex algebra language but it is not related to the light cone coordinates in physics. The notation in physics for Poisson brackets uses the variable $\sigma$ that we interpreted as  $z=e^{i\sigma}$. } 

\subsection{The double}
See \cite{hull2009double} for a physics introduction to  double field theory.  In the Hamiltonain formalism on the torus $T^{*}\mathbb{T}^{3}$, case $k=j=0$, we have fields $x_{i}(z)$ and $p_{i}(z)$ which are related to the current algebra by $\partial _{z}x_{i}(z)=\alpha_{i}(z)$ and $p_{i}(z)=\beta_{i}(z)$. For these fields, T-duality leads to consider transformations where $p_{i}(z)$ and $\partial_{z}x_{i}(z)$ will change their roles, therefore one ends up considering fields $y_{i}(z)$ such that $\partial_{z}y_{i}(z)=p_{i}(z)$. Now, we could instead of $T^{*}\mathbb{T}^{3}$ consider the \emph{double}  $\mathbb{T}^{6}$, a space where both fields are considered $(x_{i}(z),y_{i}(z))$ at the same time. This space realized the T-duality and other properties naturally.\footnote{The fields $y_{i}(z)$ have also an interpretation in CFT, for example the zero mode of the field $y_{i}(z)$ gives us the operator that jumps between the windings lattice. } In particular, the double for nilmanifolds $N(k)$ with $H_{j}$ flux was studied from this point of view in \cite{reid2009flux}. Mathematically, we can describe these doubles as torus fibrations  

\begin{equation*}
\xymatrix @C-=0.4cm @R-=0.4cm{
\mathbb{T}^{3}\ar[r]&{M(k,j)}\ar[d]\\
S^{1}\ar[r]&{N(k)}\ar[d]\\
&\mathbb{T}^{2}} \quad\quad \xymatrix { \\
\text{for}\quad k\in H^{2}(\mathbb{T}^{2},\mathbb{Z})\simeq\mathbb{Z}\quad \text{and}\quad j\in H^{3}(N(k),\mathbb{Z})\simeq\mathbb{Z}.\\}
\end{equation*} 

 The fields $x_{i}(z)$ and $y_{i}(z)$ on the double $M(k,j)$ are given by logarithmic fields as we explain in the section \ref{3333}. Algebraically, $x_{i}(z)$ and $y_{i}(z)$ are fields in
 $$ End(\mathcal{H}(k,j))[[z,z^{-1}]][\log z].$$
 In this work we are interested in the singularities between these logarithmic fields, these singularities are restricted by the relation between the logarithmic fields and the fields of the affine Kac Moody vertex algebra $V^{1}(\mathfrak{h}_{k,j})$.
 



\section{The double and the logarithmic fields for nilmanifolds}\label{3333}

In this section we introduce some definition and objects; we assume knowdledge of vertex algebras. For a introduction to vertex algebras see \cite{kacvertex}, \cite{frenkel2004vertex} and \cite{lepowsky2004introduction}.

\subsection{The courant bracket and the double for nilmanifolds}{\label{section 2.0}}

For any smooth manifold $N$ the tanget space $TN$ forms a Lie algebra considering the commutator of the vector fields. In \cite{dorfman1987dirac}, it was shown how to extend this algebraic structure to a bilinear bracket on $T N \oplus T^{*}N$ which does not quite satisfy the Lie algebra axioms. The Dorfman bracket and more generally the twisted Dorfman bracket for a choice of {\small $H\in\Gamma(\wedge^{3}T^{*}N)$}, a closed three form,  is defined as  
\begin{equation} \label{eq:courant-1}
[X+\xi, Y+\eta]_{H}=[X,Y]_{Lie}+L_{X}\eta-i_{Y}d\xi+i_{Y}i_{X}H, \quad\quad X+\xi,\, Y+\eta \in \Gamma(TN\oplus T^{*}N).
\end{equation} 
Whenever there exist global orthonormal framings for the symmetric pairing 
\begin{equation}
\braket{X+\xi, Y+\eta}=\eta(X)+\xi(Y),
\label{0}
\end{equation}
in $(TN \oplus T^*N)\otimes \mathbb{C}$ and linearly closed under the bracket \eqref{eq:courant-1}, we obtain a global trivialization $(T N \oplus T^*N) \simeq \fh \times N$ and the bracket \eqref{eq:courant-1} endows $\fh$ with a Lie algebra structure. \\

 We consider compact nilmanifolds $N = G/\Gamma$ expressed as a quotient of the nilpotent Lie group $G$ by a co-compact lattice $\Gamma \subset G$. We choose a three form $H$, called \emph{H-flux} in physics literature,  representing a class in $ H^{3}(N,\mathbb{Z})$. In this situation the tangent bundle of $N$ is naturally trivialized as $\mathfrak{g} \times N$, where $\mathfrak{g} = Lie (G)$. There exist global framings, a basis of $\mathfrak{h}$, with the above properties, and the Lie algebra $\fh$ fits into a short exact sequence
\begin{equation}
0 \rightarrow \mathfrak{g}^{*} \rightarrow \mathfrak{h} \rightarrow \mathfrak{g} \rightarrow 0.
\label{h}
\end{equation}
The class of this extension is parametrized by $H$ viewed as a map $\wedge^{2}\mathfrak{g}\rightarrow \mathfrak{g}^{*}$. The exact sequence \eqref{h} integrates, for $H\in H^{3}(N, \mathbb{Z})$, to Lie groups and co-compact discrete subgroups as 
\begin{equation}
0 \rightarrow G^{*} \rightarrow \mathfrak{H} \rightarrow G \rightarrow 0, \quad\quad  0 \rightarrow \Gamma^{*} \rightarrow \Lambda \rightarrow \Gamma \rightarrow 0,\quad\quad  \Gamma^{*}\subset  G^{*},\quad \Lambda \subset\mathfrak{H},\quad \Gamma\subset G.
\label{h2}
\end{equation}
Then we define the double $M$ as a torus bundle
\begin{equation*}
\xymatrix @C-=0.4cm @R-=0.4cm{
    \mathbb{T}^{dimN}=G^{*}/ \Gamma^{*} \ar[r]& M:=\mathfrak{H}/\Lambda \ar[d] \\
& N=G/\Gamma . }
\end{equation*}

 \subsection{Logarithmic fields}\label{section 2.3}
 
The current algebra in \cite{alekseev2005current} form by global sections of $(T N \oplus T^*N) \simeq \fh \times N$ for a nilmanifold $N=G/\Gamma$ is given by a affine Kac Moody vertex algebra
 \begin{equation} V^{1}(\mathfrak{h})=Ind^{\hat{\mathfrak{h}}}_{\mathfrak{h}[[t]]\oplus \mathbb{C}K}\ket{0}=U(\hat{\mathfrak{h}})\otimes_{U(\mathfrak{h}[[t]]\oplus\mathbb{C}K)}\mathbb{C}\ket{0},\label{eq2.3.-1}\end{equation}
 where $\mathfrak{h}$ is the Lie algebra (\ref{eq:courant-1}) and the bi-invariant pairing is given by (\ref{0}).
 
 We are interested in considering more general fields than the fields of the vertex algebra $V^{1}(\mathfrak{h})$. In order to introduce these fields we must work in a larger space of states than (\ref{eq2.3.-1}). In particular, we consider 
 \begin{equation}\mathcal{H} = \mathrm{Ind}^{\hat{\mathfrak{h}}}_{\mathfrak{h}[[t]]\oplus \mathbb{C}K}L^{2}(M) \simeq U(\hat{\fh}) \otimes_{U(\mathfrak{h}[[t]]\oplus \mathbb{C}K)} L^2(M).
\label{in.1}
\end{equation}
Where we consider the induced action from the Lie group $\mathfrak{H}$ (and therefore its Lie algebra $ \mathfrak{h}$) on $L^{2}(M)$ by right translations, hence it arises the infinite dimensional $\hat{\fh}$-module. We notice that $L^2(M)$ is a completion of the group algebra $\mathbb{C}[\Lambda]$ by Fourier expansion. 

In this work we associate coordinates to logarithmic fields as follows. First, the nilmanifold $M$ comes with a global framing on its cotangent bundle given by $T^{*}M\approx \mathfrak{h}^{*}\times M$. Then we express with coordinates $(x_{i})$ on $M$ a basis $(\alpha_{i})$ of $\mathfrak{h}^{*}$, and we assume that these are the relations satisfied by the fields  
 \begin{equation}\alpha_{i}=f(x_{i})dx_{i}\quad \Rightarrow \quad \alpha_{i}(z)=f(x_{i}(z))\partial _{z}x_{i}(z).\label{fie}
 \end{equation}
 
 The motivation behind this comes from the Hamiltonian formalism, the current algebra for $N$ relates $(u_{i},\alpha_{i})\in TN\oplus T^{*}N$ with fields $(u_{i}(z)p_{i}(z),\alpha_{i}(z)\partial_{z} x_{i}(z))$. The elements that accompany $\partial_{z} x_{i}(z)$ come from sections on the cotangent bundle on N. On the double $M$ this corresponds to $(\ref{fie})$.  
 
 Second, we express the fields $x_{i}(z)$ in modes such that they satisfy (\ref{fie}). Therefore we must consider fields also expanded by $\log(z)$ terms and their modes will be operators on $\mathcal{H}$. Then we will have fields on   
 \begin{equation}End(\mathcal{H})[[z,z^{-1}]][\log z].\label{eq2.3-2}\end{equation}

 Finally, the algebraic relations for $x_{i}(z)$ are not necessarily expressed by singularities given by delta functions and their derivatives, because we are working with logarithmic fields. In general, we do not know in advance what kind of singularities we could have between the $x_{i}(z)$ fields. In this work we study the singularities that we have for a particular class of nilmanifolds. Note that the singularities for the logarithmic fields are restricted by (\ref{fie}) given that this imposes relations with the affine Kac Moody algebra $V^{1}(\mathfrak{h})$. 
 
 \begin{example}[The torus case, $k=j=0$] \label{e1}
 In this case $\mathfrak{h}$ is abelian, then $M=\mathbb{T}^{6}$ and (\ref{fie}) gives us the simple relations
 \begin{equation}\alpha_{i}=dx_{i},\quad \beta_{i}=dy_{i}\quad\Rightarrow\quad \partial_{z} x_{i}(z)=\alpha_{i}(z),\quad\partial_{z} y_{i}=\beta_{i}(z).\label{ee1}\end{equation}

The more general fields which satisfies the relation above can be expressed by the logarithmic fields 
$$x_{i}(z)=w_{i}\log z+\sum_{n\in\mathbb{Z}}x^{i}_{n}z^{-n},\quad\quad  y_{i}(z)=p_{i}\log z+\sum_{n\in\mathbb{Z}}y^{i}_{n}z^{-n}.$$
These fields satisfy the logarithmic singularity
$$[x_{i}(z),y_{j}(w)]=log(z-w),$$
we defined this singularity in (\ref{log}). 
 \end{example}

 \begin{example}[The \emph{double twisted torus} case $k=0$, $j=1$]\label{e2}
 In this case $\mathfrak{h}$ is not abelian but 2 step nilpotent, and M is a $\mathbb{T}^{3}$-fibration over $\mathbb{T}^{3}$, (\ref{fie}) gives us the relations
$$\alpha_{i}=dx_{i},\quad \beta_{i}=dy_{i}-\frac{1}{2}\varepsilon_{ijk}x_{j}dx_{k}
\quad\Rightarrow\quad \partial _{z}x_{i}(z)=\alpha_{i}(z), \quad\partial_{z}y_{i}(z)=\beta_{i}(z)+\frac{1}{2}\varepsilon_{ijk}:x_{j}(z)\partial_{z}x_{k}(z){ :}.$$

The more general fields which satisfy the relation above can be expressed by the logarithmic fields 
$$x_{i}(z)=w_{i}\log z+\sum_{n\in\mathbb{Z}} x^{i}_{n}z^{-n},\quad\quad y_{i}(z)=p_{i}\log z+\sum_{n\in\mathbb{Z}} y^{i}_{n}z^{-n}+\frac{\varepsilon_{ijk}}{2}w_{j}x_{k}(z)\log z.
$$

These fields satisfy the following algebraic relations 
\begin{equation}
\begin{split}
&[x_{i}(z),x_{j}(w)]=0,\quad\quad [x_{i}(z),y_{j}(w)]=\delta_{ij}\log(z-w),\\
&[y_{i}(z),y_{j}(w)]=\varepsilon_{ijk}w_{k}rl(z,w)+\frac{1}{2}\varepsilon_{ijk}(\hat{x}_{k}(z)-\hat{x}_{k}(w))\log (z-w).
\end{split}
\label{rll}
\end{equation}

Where we used the singularity (\ref{5}) and the notation $\hat{x}_{i}(z)$ means the projection onto fields without $\log z$ terms, that is $\hat{.}:End (\mathcal{H})[[z,z^{-1}]][\log z]\rightarrow  End(\mathcal{H})[[z,z^{-1}]]$.
\end{example}

For the general case $k\neq 0$ and $j\neq 0$, we will express the algebraic relation between the modes $\{w_{i}, p_{i}, x^{i}_{n}, y^{i}_{n}\}$ for ${n\in\mathbb{Z}}$ and $i\in\{1,2,3\}$. For the particular examples before the algebra of these modes is a Lie algebra, but for the general case the modes will form a nonlinear Lie algebra. 

\subsection{Non-linear Lie algebras}\label{Non linear}

In this section we follow \cite{de2006finite}. Let $\mathfrak{g}$ a vector space, and let $T(\mathfrak{g})$ denote the tensor algebra over $\mathfrak{g}$. If $\mathfrak{g}$ is endowed with a linear map 

$$[,]:\mathfrak{g}\otimes \mathfrak{g}\rightarrow T(\mathfrak{g}).$$

We extended it to $T(\mathfrak{g})$ by the Leibnitz rule, for $A, B\in T(\mathfrak{g})$
$$[A\otimes B, C]=[A,C]\otimes B+ A\otimes [B,C],\quad\quad [A, B\otimes C]=[A,B]\otimes C+B\otimes [A,C].$$

We define 
$$\mathcal{M}(\mathfrak{g}):=span\{A\otimes (b\otimes c-c\otimes b-[b,c])\otimes D |b,c \in \mathfrak{g}, A, D\in T(\mathfrak{g}) \}.$$
Note that $\mathcal{M}(\mathfrak{g})$ is the two sided ideal of the tensor product $T(\mathfrak{g})$ generate by elements $(b\otimes c-c\otimes b-[b,c])$, where $a,b \in \mathfrak{g}$.

\begin{defn} A non linear Lie algebra $\mathfrak{g}$ is a vector space with a linear map $[,]:\mathfrak{g}\otimes \mathfrak{g}\rightarrow T(\mathfrak{g})$ satisfying the following properties $(a,b,c\in \mathfrak{g})$
\begin{itemize}
    \item skewsymmetry: $[a,b]=-[b,a]$.
    \item $[a,[b,c]]-[b,[a,c]]-[[a,b],c]\in\mathcal{M}(\mathfrak{g})$.
\end{itemize}
\end{defn}
The associative algebra $U(\mathfrak{g})=T(\mathfrak{g})/\mathcal{M}(\mathfrak{g})$ is called the universal enveloping algebra of the non linear Lie algebra $\mathfrak{g}$





\section{Logarithmic fields and their singularities, case $k\neq 0$ and $j\neq 0 $}\label{chap3}

In this section we describe the double $M(k,j)$ and their logarithmic fields. We express the singularities between these fields, and the algebra that emerges from their modes.
  
\subsection{The double $M(k,j)$ and their logarithmic fields $x_{i}(z)$ and $y_{i}(z)$} \label{sec:3.1}

Let $k \in \mathbb{Z}$ and $G_k$ be the $3$--dimensional Heisenberg group. It is the manifold $G_k=\mathbb{R}^3$ with multiplication:
$$(x,y,z)(x',y',z')=\left(x+x',y+y'-\frac{k}{2}xz'+\frac{k}{2}x'z,z+z'\right).$$

Let $\Gamma \subset G_k$ be the subgroup generated by the standard basis of $\mathbb{R}^3$. It is a co-compact lattice. We have the corresponding nilmanifold, usually referred to as the \emph{Heisenberg nilmanifold} $N(k):=G_k/\Gamma$. Notice that for all $k$ the groups $G_k$ are isomorphic, but under these isomorphisms, the corresponding $\Gamma$ are not intertwined.

Let $j \in \mathbb{Z} \simeq H^3(N, \mathbb{Z})$ and consider a three form  $H_{j}=-jdx\wedge dy\wedge dz$ representing this class. In particular, sections of the bundle $TN(k) \oplus T^*N(k)$ are endowed with a bilinear operation (the $H$-twisted Dorfman bracket (\ref{eq:courant-1})). This bundle admits a global framing with their respective brackets given in (\ref{0.0}). Hence we obtain the $6$ dimensional three step nilpotent Lie algebra $\fh_{k,j}$ and the trivialization $TN(k) \oplus T^*N(k) \simeq \fh_{k,j} \times N$ . Notice that $\fh_{k,j}$ fits into a short exact sequence as in (\ref{h}), and this extension is an abelian extension but it is not a central extension if $k \neq 0$. 


Let $\mathfrak{H}_{k,j}$ be the nilpotent Lie group with Lie algebra $\fh_{k,j}$. As a manifold it is $\mathbb{R}^6$, its multiplication table can be found by the BCH formula:
\begin{align}
\begin{split} 
&\left( x_{1},x_{2},x_{3},y_{1},y_{2},y_{3} )( x^{\ast}_{1},x_{2}^{\ast},x_{3}^{\ast},y_{1}^{\ast},y_{2}^{\ast},y_{3}^{\ast}\right) =( x_{1}^{\ast\ast},x_{2}^{\ast\ast},x_{3}^{\ast\ast},y_{1}^{\ast\ast},y_{2}^{\ast\ast},y_{3}^{\ast\ast} )\\
& x_{1}^{\ast\ast}=x_{1}+x_{1}^{\ast},\quad\quad x_{3}^{\ast\ast}=x_{3}+x_{3}^{\ast},\quad\quad x_{2}^{\ast\ast}=x_{2}+x_{2}^{\ast}+\frac{k}{2}(x_{3}x_{1}^{\ast}-x_{1}x_{3}^{\ast}),\\
&y_{1}^{\ast\ast}=y_{1}+y_{1}^{\ast}+\frac{k}{2}(y_{2}x_{3}^{\ast}-y_{2}^{\ast}x_{3})+\frac{j}{2}(x_{2}x_{3}^{\ast}-x_{2}^{\ast}x_{3})+\frac{kj}{6}(x_{3}^{\ast}-x_{3})(x_{3}x_{1}^{\ast}-x_{1}x_{3}^{\ast}),\\
& y_{2}^{\ast\ast}=y_{2}+y_{2}^{\ast}+\frac{j}{2}(x_{3}x_{1}^{\ast}-x_{1}^{\ast}x_{3}),\\
&y_{3}^{\ast\ast}=y_{3}+y_{3}^{\ast}+\frac{k}{2}(x_{1}y_{2}^{\ast}-x_{1}^{\ast}y_{2})+\frac
{j}{2}(x_{1}x_{2}^{\ast}-x_{1}^{\ast}x_{2})+\frac{kj}{6}(x_{1}-x_{1}^{\ast})(x_{3}x_{1}^{\ast}-x_{1}x_{3}^{\ast}).
\end{split}
\label{bch}
\end{align}

The canonical basis of $\mathbb{R}^6$ generates a co-compact lattice $\Lambda_{k,j} \subset \mathfrak{H}_{k,j}$. The quotient 
$M= M(k,j)=\mathfrak{H}_{k,j}/\Lambda_{k,j}$ is a compact $6$-dimensional nilmanifold and it is a $\mathbb{T}^{3}$ bundle on $N(k)$ as we described in the introduction and in the section \ref{section 2.0}.  

\begin{remark} It turns out that $M(k,j)=M(j,k)$, a phenomenon which can be explained by topological T-duality. In \cite{hull2009double} is given a physical interpretation of the double. In \cite{bouwknegt2004t} a similar construction is given but instead of work with $M(k,j)$ they consider ${N(k)\times_{\mathbb{T}^{2}}N(j)}$. Finally, we consider nilmanifolds and its nilpotent Lie algebras, a construction of the double for simple Lie groups and algebras is given in \cite{vsevera2015poisson}.
\end{remark}

$M(k,j)$ is a nilmanifold therefore its cotangent bundle can be trivialized with left invariant forms given by:
\begin{align} 
\begin{split}
\alpha_{1}&=dx_{1},\quad\quad\alpha_{2}=dx_{2}-\frac{1}{2}kx_{3}dx_{1}+\frac{1}{2}kx_{1}dx_{3},\quad\quad \alpha_{3}=dx_{3},\\
\beta_{1}&=dy_{1}-\frac{kj}{3}x_{3}^{2}dx_{1}+\frac{1}{2}jx_{3}dx_{2}-\frac{1}{2}(ky_{2}+jx_{2})dx_{3}+\frac{kj}{3}x_{1}x_{3}dx_{3}+\frac{1}{2}kx_{3}dy_{2},\\
\beta_{2}&=dy_{2}-\frac{1}{2}jx_{3}dx_{1}+\frac{1}{2}jx_{1}dx_{3},\\
\beta_{3}&=dy_{3}+\frac{kj}{3}x_{3}x_{1}dx_{1}+\frac{1}{2}(ky_{2}+jx_{2})dx_{1}-\frac{1}{2}jx_{1}dx_{2}-\frac{kj}{3}x_{1}^{2}dx_{3}-\frac{1}{2}kx_{1}dy_{2}.
\end{split}
 \label{eq:formas-en-m}
\end{align}

And (\ref{fie}) gives us the relations
{\small 
 \begin{align}
 \begin{split}
 &\alpha_{1}(z)=\partial x_{1}(z),\quad\quad \alpha_{2}(z)=\partial x_{2}(z)-\frac{1}{2}kx_{3}(z)\partial x_{1}(z)+\frac{1}{2}kx_{1}(z)\partial x_{3}(z),\quad\quad \alpha_{3}(z)=\partial x_{3}(z),\\
  &\beta_{1}(z)=\partial y_{1}(z)+\frac{1}{2}j(x_{3}(z)\partial x_{2}(z)-x_{2}(z)\partial x_{3}(z))+\frac{1}{2}k(x_{3}(z)\partial y_{2}(z)-y_{2}(z)\partial x_{3}(z))\\
 &\quad\quad\quad\quad\quad\quad\quad\quad\quad\quad\quad\quad\quad\quad\quad\quad\quad\quad\quad\quad\quad\quad\quad\quad\quad-\frac{kj}{3}x_{3}^{2}(z)\partial x_{1}(z)+\frac{kj}{3}x_{3}(z)x_{1}(z)\partial x_{3}(z),\\
 &\beta_{2}(z)=\partial y_{2}(z)-\frac{1}{2}jx_{3}(z)\partial x_{1}(z)+\frac{1}{2}jx_{1}(z)\partial x_{3}(z),\\
  &\beta_{3}(z)=\partial y_{3}(z)+\frac{1}{2}j(x_{2}(z)\partial x_{1}(z)-x_{1}(z)\partial x_{2}(z)+\frac{1}{2}k(y_{2}(z)\partial x_{1}(z)-x_{1}(z)\partial y_{2}(z))\\
  &\quad\quad\quad\quad\quad\quad\quad\quad\quad\quad\quad\quad\quad\quad\quad\quad\quad\quad\quad\quad\quad\quad\quad\quad\quad-\frac{kj}{3}x_{1}^{2}(z)\partial x_{3}(z)+\frac{kj}{3}x_{3}(z)x_{1}(z)\partial x_{1}(z).\\
 \end{split}
 \label{formas}
 \end{align}}

The more general fields which satisfy the relation above can be expressed by the following logarithmic fields 
 {\small
\begin{align}
\begin{split}
&x_{1}(z)=w_{1}\log z+\sum x^{1}_{n}z^{-n},\\
&x_{3}(z)=w_{3}\log z+\sum x^{3}_{n}z^{-n},\\
&x_{2}(z)=w_{2}\log z+\sum x_{n}^{2}z^{-n}+\frac{1}{2}k\log z(w_{3}x_{1}(z)-w_{1}x_{3}(z)),\\
   &y_{2}(z)=p_{2}\log z+\sum y_{n}^{2}z^{-n}+\frac{1}{2}j\log z(w_{3}x_{1}(z)-w_{1}x_{3}(z)),\\
   &y_{1}(z)=p_{1}\log z+\sum y_{n}^{1}z^{-n}+1/2k\log z(p_{2}x_{3}(z)-w_{3}y_{2}(z))+j/2\log z(w_{2}x_{3}(z)-w_{3}x_{2}(z))\\
&\quad\quad\quad\quad+\frac{kj}{6}w_{3}(\log z)^{2}(w_{3}x_{1}(z)-w_{1}x_{3}(z))+\frac{kj}{6}x_{3}(z)\log z(w_{3}x_{1}(z)-w_{1}x_{3}(z)),\\
&y_{3}(z)=p_{3}\log z+\sum y_{n}^{3}z^{-n}+1/2k\log z(w_{1}y_{2}(z)-p_{2}x_{1}(z))+j/2\log z(w_{1}x_{2}(z)-w_{2}x_{1}(z))\\
&\quad\quad\quad\quad-\frac{kj}{6}w_{1}(\log z)^{2}(w_{3}x_{1}(z)-w_{1}x_{3}(z))-\frac{kj}{6}x_{1}(z)\log z(w_{3}x_{1}(z)-w_{1}x_{3}(z)).\\
\end{split}
\label{modos}
\end{align}  
}

\subsection{Algebraic relations between the logarithmic fields $x_{i}(z)$ and $y_{i}(z)$}

For arbitrary $k$ and $j$ the algebraic relations are substantially more complicated than the previous cases, examples \ref{e1} and \ref{e2}. 

\begin{teo}\label{principal}
 The following commutation relations for the fields $x_{i}(z)$ and $y_{i}(z)$ in \eqref{modos} imply the commutation relations in \eqref{1}, of the affine Kac Moody vertex algebra $V^{1}(\mathfrak{h}_{k,j})$, for the fields $\alpha_{i}(z)$ and $\beta_{i}(z)$ in \eqref{formas}.
 \begin{align*}
 &[x_{i}(z),y_{j}(w)]=\delta_{ij}\log (z-w),
 \\
 &[y_{1}(z),y_{2}(w)]=\frac{j}{2}(\hat{x}^{3}(z)-\hat{x}^{3}(w))\log (z-w)+jw_{3}rl(w,z),
 \\
 &[y_{1}(z),x_{2}(w)]=\frac{k}{2}(\hat{x}^{3}(z)-\hat{x}^{3}(w))\log (z-w)+kw_{3}rl(w,z),
 \\
 &[y_{2}(z),y_{3}(w)]=\frac{j}{2}(\hat{x}^{1}(z)-\hat{x}^{1}(w))\log (z-w)+jw_{1}rl(w,z),
 \\
 &[x_{2}(z),y_{3}(w)]=\frac{k}{2}(\hat{x}^{1}(z)-\hat{x}^{1}(w))\log (z-w)+kw_{1}rl(w,z),
\end{align*}
\begin{align*}
[y_{1}(z),y_{1}(w)]&=-\frac{kj}{6}\left(\hat{x}_{3}(z)^{2}+\hat{x}_{3}(w)^{2}-3\hat{x}_{3}(z)\hat{x}_{3}(w) \right) {\log(z-w)}\nonumber\\
 &+\frac{kj}{6}w_{3}\left(\hat{x}_{3}(w)\log w+\hat x_{3}(z)\log z\right) {{\log(z-w)}}\nonumber\\  
&+kjw_{3}\left(\hat{x}_{3}(w)-\hat{x}_{3}(z)\right)rl(z,w)+kjw_{3}w_{3}t(z,w), 
\end{align*}
\begin{align*}
[y_{1}(z),y_{3}(w)]&=\frac{k}{2}\left(\hat{y}_{2}(w)-\hat{y}_{2}(z)\right){{\log(z-w)}}+\frac{j}{2}\left(\hat{x}_{2}(w)-\hat{x}_{2}(z)\right){{\log(z-w)}}\nonumber\\
& +\frac{kj}{6}\left(\hat{x}_{3}(z)\hat{x}_{1}(z)+\hat{x}_{1}(w)\hat{x}_{3}(w)-3\hat{x}_{3}(z)\hat{x}_{1}(w)\right){{log(z-w)}}\nonumber               \\
      & +\frac{kj}{6}\left(w_{3}\hat{x}_{1}(w)\log w-2w_{3} \hat x_{1}(z)\log z+w_{1}\hat{x}_{3}(z)\log z-2w_{1}\hat{x}_{3}(w)\log w\right) {{\log(z-w)}}\nonumber\\  
      &+kj\left(w_{1}\hat{x}_{3}(z)-w_{3}\hat{x}_{1}(w)\right){{rl(z,w)}}-(jw_{2}+kp_{2}){{rl(z,w)}}-kjw_{3}w_{1}{{t(z,w)}},
\end{align*}
\begin{align*} [y_{3}(z),y_{3}(w)]&=-\frac{kj}{6}\left(\hat{x}_{1}(z)^{2}+\hat{x}_{1}(w)^{2}-3\hat{x}_{1}(z)\hat{x}_{1}(w)\right){{\log(z-w)}}\nonumber\\
&+\frac{kj}{6}w_{1}\left(\hat{x}_{1}(w)\log w+\hat{x}_{1}(z)\log z\right){{\log(z-w)}}\nonumber\\
&+kjw_{1}\left(\hat{x}_{1}(w)-\hat{x}_{1}(z)\right){{rl(z,w)}}+kjw_{1}w_{1}{{t(z,w)}},
\end{align*}
where the function $t(z,w)$ is the function defined in (\ref{5.5}).
\label{thm:3.3.1}
\end{teo}

The proof of this theorem is given in the appendix \ref{proof}. 



The modes, when expressed as operators acting on $\cH(k,j)$, are not linearly closed under commutators. In fact, they form a non-linear Lie algebra, see section \ref{Non linear}. The commutation relations between these modes, while not a linear combination of themselves, is compatible with the Jacobi identity.

\begin{teo}\label{theorem4.2}For each pair of integer numbers $k$ and $j$ exists a non-linear Lie algebra with basis $ \left\{ x^i_n, y^i_n, w_i, p_i \right\}$, $i=1,2,3$, $n \in \mathbb{Z}$; the quadratic commutation relations are given by  
\begin{equation*} 
\begin{gathered}
\begin{aligned}
{[}x^{i}_{n},y^{j}_{m}]&=\delta_{ij}\frac{\delta_{n,-m}}{m}, & \\
{[}y^{1}_{n},y^{2}_{m}]&=\frac{j}{2}x^{3}_{n+m}\left(\frac{1}{n}+\frac{1}{m}\right)+jw_{3}\frac{\delta_{n,-m}}{m^{2}},\quad &
[y^{1}_{n},x^{2}_{m}]&=\frac{k}{2}x^{3}_{n+m}\left(\frac{1}{n}+\frac{1}{m}\right)+kw_{3}\frac{\delta_{n,-m}}{m^{2}}, \\
[y^{2}_{n},y^{3}_{m}]&=\frac{j}{2}x^{1}_{n+m}\left(\frac{1}{n}+\frac{1}{m}\right)+jw_{1}\frac{\delta_{n,-m}}{m^{2}},\quad &
{[}x^{2}_{n},y^{3}_{m}]&=\frac{k}{2}x^{1}_{n+m}\left(\frac{1}{n}+\frac{1}{m}\right)+kw_{1}\frac{\delta_{n,-m}}{m^{2}}, 
\end{aligned}
\end{gathered}
\end{equation*}
{\small\begin{equation*}
\begin{gathered}
\begin{aligned}
&[y^{1}_{n},y^{1}_{m}]=-2kjw_{3}w_{3}\frac{\delta_{n,-m}}{m^{3}}-kj\left( \frac{1}{m^{2}}-\frac{1}{n^{2}}\right)w_{3}x^{3}_{n+m} +\frac{kj}{2}\sum_{l}\frac{x^{3}_{n+l}x^{3}_{m-l}}{l}-\frac{kj}{6}\sum_{l}x^{3}_{l}x^{3}_{m+n-l}\left(\frac{1}{m}-\frac{1}{n}\right), \\
&[y^{1}_{n},y^{3}_{m}]=-\frac{1}{2}\left(ky^{2}_{n+m}+jx^{2}_{n+m}\right)\left(\frac{1}{n}+\frac{1}{m}\right)-\left(kp_{2}+jw_{2}\right)\frac{\delta_{n,-m}}{m^{2}} + \\
&\quad\quad 2kjw_{1}w_{3}\frac{\delta_{n,-m}}{m^{3}}+kj\left(\frac{1}{m^{2}}w_{1}x^{3}_{n+m}-\frac{1}{n^{2}}w_{3}x^{1}_{n+m}\right)-\frac{kj}{2}\sum_{l}\frac{x^{3}_{n+l}x^{1}_{m-l}}{l}+\frac{kj}{6}\sum_{l}x^{1}_{l}x^{3}_{m+n-l}\left(\frac{1}{m}-\frac{1}{n}\right),&&& \\
&[y^{3}_{n},y^{3}_{m}]=-2kjw_{1}w_{1}\frac{\delta_{n,-m}}{m^{3}}-kj\left(\frac{1}{m^{2}}-\frac{1}{n^{2}}\right)w_{1}x^{1}_{n+m}+\frac{kj}{2}\sum_{l}\frac{x^{1}_{n+l}x^{1}_{m-l}}{l}-\frac{kj}{6}\sum_{l}x^{1}_{l}x^{1}_{m+n-l}\left(\frac{1}{m}-\frac{1}{n}\right),
\end{aligned}
\end{gathered}
\end{equation*} }
\begin{equation} 
\begin{aligned}
&[w_{i},y^{j}_{m}]=\delta_{ij}\delta_{0,m},\quad\quad [x^{i}_{n},p_{j}]=-\delta_{ij}\delta_{n,0},\quad\quad [w_{i},p_{j}]= [p_{1},p_{1}]=[p_{3},p_{3}]=0,\\
&[p_{1},y^{2}_{m}]=\frac{j}{2}x^{3}_{m},\quad [y^{1}_{n},p_{2}]=\frac{j}{2}x^{3}_{n},\quad [p_{1},x^{2}_{m}]=\frac{k}{2}x^{3}_{m},\quad [y^{1}_{n},w_{2}]=\frac{k}{2}x^{3}_{n},\\
&[p_{2},y^{3}_{m}]=\frac{j}{2}x^{1}_{m},\quad [y^{2}_{n},p_{3}]=\frac{j}{2}x^{1}_{n},\quad [w_{2},y^{3}_{m}]=\frac{k}{2}x^{3}_{m},\quad [x^{2}_{n},p_{3}]=\frac{k}{2}x^{1}_{n},\\
& [p_{1},w_{2}]=kw_{3},\quad [p_{1},p_{2}]=jw_{3},\quad [w_{2},p_{3}]=kw_{1},\quad [p_{2},p_{3}]=jw_{1},\quad [p_{1},p_{3}]=+(jw_{2}+kp_{2}),\\
&[p_{1},y^{1}_{m}]=-\frac{kj}{6}\sum_{l}x^{3}_{l}x^{3}_{m-l},\quad\quad [y^{1}_{n},p_{1}]=\frac{kj}{6}\sum_{l}x^{3}_{l}x^{3}_{n-l},\\
&[p_{1},y^{3}_{m}]=\frac{kj}{6}\sum_{l}x^{1}_{l}x^{3}_{m-l}+\frac{1}{2}(ky^{2}_{m}+jx^{2}_{m}),\quad\quad[y^{1}_{n},p_{3}]=-\frac{kj}{6}\sum_{l}x^{1}_{l}x^{3}_{n-l}+\frac{1}{2}(ky^{2}_{n}+jx^{2}_{n}),\\
&[p_{3},y^{3}_{m}]=-\frac{kj}{6}\sum_{l}x^{1}_{l}x^{1}_{m-l},\quad\quad [y^{3}_{n},p_{3}]=\frac{kj}{6}\sum_{l}x^{1}_{l}x^{1}_{n-l},\\
\end{aligned}
\label{eq:commut-thm2}
\end{equation} 
where we understand $0$ whenever we have an expression $1/n$ for $n=0$. 
\label{thm:3.3.2}
\end{teo}

\begin{proof}
The commutation relations above follow from the fields definition \eqref{modos} and the algebraic relations given in the theorem \ref{thm:3.3.1}. The computation to check the Jacobi identity is long but straightforward, here we record the most complicated part:

\begin{align*}
[y^{1}_{\ell},[y^{1}_{n},y^{3}_{m}]]
={}&-\frac{1}{2}\left(kjx^{3}_{\ell+n+m}\left(\frac{1}{\ell}+\frac{1}{m+n}\right)+kjw_{2}\frac{\delta_{\ell,-n-m}}{\ell^{2}}\right)\left(\frac{1}{n}+\frac{1}{m}\right)-kjx^{3}_{\ell}\frac{\delta_{n,-m}}{m^{2}} \\
&-2kjw_{3}\delta_{\ell,0}\frac{\delta_{n,-m}}{m^{3}}+kj\left(-\frac{1}{m^{2}}\delta_{\ell,0}x^{3}_{n+m}-\frac{1}{n^{2}}w_{3}\frac{\delta_{\ell,-n-m}}{n+m}\right) \\
&+\frac{kj}{2}x^{3}_{n+\ell+m}\frac{1}{\left(\ell+m\right)\ell}+\frac{kj}{6}\frac{1}{-\ell}x^{3}_{m+n+\ell}\left(\frac{1}{m}-\frac{1}{n}\right) \\
[y^{3}_{m},[y^{1}_{\ell},y^{1}_{n}]]
={}&4kjw_{3}\delta_{m,0}\frac{\delta_{\ell,-n}}{n^{3}}-kj\left(\frac{1}{n^{2}}\delta_{m,0}x^{3}_{\ell+n}+\frac{1}{\ell^{2}}\delta_{m,0}x^{3}_{\ell+n}\right)-kjw_{3}\left(\frac{1}{n^{2}}+\frac{1}{\ell^{2}}\right)\frac{1}{\ell+n}\delta_{m,-\ell-n} \\
&+\frac{kj}{2}x^{3}_{n+m+\ell}\frac{1}{\left(\ell+m\right)\left(m\right)}-\frac{kj}{2}x^{3}_{\ell+m+n}\frac{1}{\left(n+m\right)\left(m\right)} \\
&-\frac{kj}{6}\left(\frac{1}{n}-\frac{1}{\ell}\right)x^{3}_{n+m+\ell}\frac{1}{-m}-\frac{kj}{6}x^{3}_{\ell+m+n}\frac{1}{-m}\left(\frac{1}{n}-\frac{1}{\ell}\right) \\
[y^{1}_{n},[y^{3}_{m},y^{1}_{\ell}]]
={}&\frac{1}{2}\left(kjx^{3}_{\ell+n+m}\left(\frac{1}{n}+\frac{1}{m+\ell}\right)+kjw_{2}\frac{\delta_{n,-\ell-m}}{n^{2}}\right)\left(\frac{1}{m}+\frac{1}{\ell}\right)+kjx^{3}_{n}\frac{\delta_{\ell,-m}}{m^{2}} \\
&+2kjw_{3}\delta_{n,0}\frac{\delta_{\ell,-m}}{m^{3}}-kj\left(-\frac{1}{m^{2}}\delta_{n,0}x^{3}_{\ell+m}-\frac{1}{\ell^{2}}w_{3}\frac{\delta_{n,-\ell-m}}{\ell+m}\right) \\
&-\frac{kj}{2}x^{3}_{n+\ell+m}\frac{1}{\left(n+m\right)n}-\frac{kj}{6}\frac{1}{-n}x^{3}_{m+n+\ell}\left(\frac{1}{m}-\frac{1}{\ell}\right)
\end{align*}

Therefore we have
$$[y^{1}_{l},[y^{1}_{n},y^{3}_{m}]]+[y^{3}_{m},[y^{1}_{l},y^{1}_{n}]]+[y^{1}_{n},[y^{3}_{m},y^{1}_{l}]]=0.$$
\end{proof}



\appendix

  \section{From the Hamiltonian formalism to CFT}\label{CFTH}
 Until now, we have only considered the Poisson brackets using the vertex algebras formalism \eqref{1}. Now we will also consider the Hamiltonian operator. The Hamiltonian is given by a field $h(z)=\sum_{n\in \mathbb{Z}}h_{n}z^{-n-2}$, more specifically, by its zero mode $h_{0}$. The equations of motion for a field $A(z)=\sum_{n\in\mathbb{Z}}a_{(n)}z^{-n-1}$ are given by
\begin{equation}\frac{d}{d\tau}A(z)= [h_{0}, A(z)]\quad \Rightarrow \quad A(z,\tau) =e^{\tau h_{0}}A(z)e^{-\tau h_{0}}=\sum_{n\in\mathbb{Z}}a_{(n)}(\tau)z^{-n-1}.
\label{555}
\end{equation}

Because of $A(z,0)=A(z)$ and $B(z,0)=B(z)$, we know the algebraic relations in $\tau=0$ of the fields $A(z,\tau)$ and $B(z,\tau)$. On the other hand, for an arbitrary $\tau$ the algebra depends of the equations of motion that in general could be hard to solve. 

Now, it could happen that for some cases the theory satisfies extra properties. For example, the beta equations in a CFT. In this case, the fields are described into two \emph{chiral} parts\footnote{Physically this comes from the quantization of the conformal symmetry, that gives us the two non-vanishing components of the energy momentum tensor.} which depend on $\zeta=e^{\tau+i\sigma}$ and $\bar{\zeta}=e^{\tau-i\sigma}$. This means, we will have fields $C(z)=\sum_{n\in\mathbb{Z}}c_{(n)}z^{-n-1}$ and $D(z)=\sum_{n\in\mathbb{Z}}d_{(n)}z^{-n-1}$ such that 
{\fontsize{10}{11}\begin{equation}
\begin{cases}
 C(\zeta):=C(z,\tau)=\sum_{n\in\mathbb{Z}}c_{(n)}(\tau)z^{-n-1}=\sum_{n\in\mathbb{Z}}c_{(n)}\zeta^{-n-1}\\
 D(\zeta):=D(z,\tau)=\sum_{n\in\mathbb{Z}}d_{(n)}(\tau)z^{-n-1}=\sum_{n\in\mathbb{Z}}d_{(-n)}\bar{\zeta}^{-n-1}\\
\end{cases}
\text{ and }\quad 
 [C(\zeta),D(\bar{\zeta})]=0,
 \label{4}
\end{equation}}
the modes of the fields are diagonal for the Hamiltonian. Note that we know the brackets between these fields because we already know the algebra of their modes.

\subsection{Torus case}

On the $\mathbb{T}^{3}$ torus case, we can consider the Hamiltonian\footnote{The Hamiltonian could be more general considering a flat metric $G_{ij}$, but here we are considering the simplest case $G_{ij}=\delta_{ij}$} as the zero mode of the field $h(z)=\partial_{z}x_{i}(z)\partial_{z}x_{i}(z)+\partial_{z} y_{i}(z)\partial _{z}y_{i}(z)$. In this case, we have a CFT. We are interested in the particular fields\footnote{We used the fields $y_{i}(z)$ on the double torus, this matches with the standard description by the relation $\partial_{z}y_{i}(z)=p_{i}(z)$. } $C(z):=\partial_{z}y_{i}(z)+\partial_{z}x_{i}(z)$ and $D(z):=\partial_{z}y_{i}(z)-\partial_{z}x_{i}(z)$, they satisfy the conditions in (\ref{4}). In particular, we can use these fields to describe the fields $x(z,\tau)=x(\zeta,\bar{\zeta})=x'(\zeta)+x''(\bar{\zeta})$ such that $\partial_{\zeta}x'=C(\zeta)$ and $\partial_{\bar{\zeta}}x''=D(\bar{\zeta})$ using the notation in (\ref{2}), we have 
{\fontsize{10}{11}\begin{equation*}
\begin{cases}
[x'_{i}(\zeta), x'_{j}(\omega)]=\delta_{ij}2{ \log (\zeta-\omega)}\\
[ x'_{i}(\zeta), x''_{j}(\bar{\omega})]=0\\
[ x''_{i}(\bar{\zeta}), x''_{j}(\bar{\omega})]=-\delta_{ij}2{\log (\bar{\zeta}-\bar{\omega})}
\end{cases}
,
\begin{cases}
[\partial_{\zeta} x_{i}(\zeta),\partial_{\omega} x_{i}(\omega)]=\delta_{ij}2\partial_{\omega}\delta(\zeta-\omega)\\
[\partial_{\zeta} x_{i}(\zeta),\partial_{\bar{\omega}} x_{j}(\bar{\omega})]=0\\
[\partial_{\bar{\zeta}} x_{i}(\bar{\zeta}),\partial_{\bar{\omega}} x_{j}(\bar{\omega})]=-\delta_{ij}2\partial_{\bar{\omega}}\delta(\bar{\zeta}-\bar{\omega})
\end{cases}, 
 \begin{cases}
T(\zeta):=\frac{1}{4}\partial_{\zeta}x_{i}(\zeta)\partial_{\zeta}x_{i}(\zeta)\\
T(\bar{\zeta}):=-\frac{1}{4}\partial_{\bar{\zeta}}x_{i}(\bar{\zeta})\partial_{\bar{\zeta}}x_{i}(\bar{\zeta})
\end{cases},
\end{equation*}}
where $T(\zeta)$ and $T(\bar{\zeta})$ are two copies of the Virasoro algebra with central charge dim$\mathbb{T}^{3}=3$.

\subsection{Twisted torus case}
Now for the twisted torus $\mathbb{T}^{3}$ with $H_{j}$ flux, the Hamiltonian is given by the zero mode of the field  $h(z)=\partial_{z}x_{i}(z)\partial_{z}x_{i}(z)+\partial_{z} y_{i}(z)\partial _{z}y_{i}(z)$. In this case, we do not have a CFT. In \cite{blumenhagen2011non} this case was studied as a perturvative CFT up to first order because the beta equations are satisfied only for this order. We are interested in the particular\footnote{We arrive at these fields from the classical theory that we have not considered here.} fields{\fontsize{10}{11}\begin{align*}
C(z):=-\beta_{i}(z)+\alpha_{i}(z)+j\epsilon_{ijk}x_{j}(z)(\beta_{k}(z)-\alpha_{k}(z)), \text {  and}\quad D(z):=\beta_{i}(z)+\alpha_{i}(z)+j\epsilon_{ijk}x_{j}(z)(\beta_{k}(z)+\alpha_{k}(z)).\end{align*} }
Where $\epsilon_{ijk}$ is the Levi Civita tensor. These fields satisfy the conditions in (\ref{4}) up to first order in $j$. Therefore, in this approach $j$ loses its topological meaning. We can use these fields to describe the field $y(z,\tau)=y(\zeta,\bar{\zeta})=y'(\zeta)+y''(\bar{\zeta})+O(j^{2})$ such that $\partial_{\zeta}y'=C(\zeta)$ and $\partial_{\bar{\zeta}}y''=D(\bar{\zeta})$. In particular, using the notation in (\ref{5}), we have  \\
{\fontsize{10}{11}
\begin{equation}
\begin{cases}
 [y'_{i}( \zeta),y'_{j}( \omega)]=-2\delta_{ij} { \log ( \zeta- \omega)}-\frac{j}{2}\epsilon_{ijk}(\hat{y}'_{k}( \zeta)-\hat{y}'_{k}( \omega)){ \log ( \zeta- \omega)}-j\kappa_{k}'{ rl( \zeta, \omega)}+{\color{black}O(j^{2})}\\
 [y'_{i}( \zeta),y''_{j}(\bar{ \omega})]=0+{\color{black}O(j^{2})}\\
 [y''_{i}(\bar{ \zeta}),y''_{j}(\bar{ \omega})]=2\delta_{ij}{ \log (\bar{ \zeta}-\bar{ \omega})}+\frac{j}{2}\epsilon_{ijk}(\hat{y}''_{k}(\bar{ \zeta})-\hat{y}''_{k}(\bar{ \omega})){ \log (\bar{ \zeta}-\bar{ \omega})}-j\kappa_{k}''{ rl(\bar{ \zeta},\bar{ \omega})}+{\color{black}O(j^{2})}
\end{cases}
\label{6}
\end{equation}
\begin{equation}
\begin{split}
\begin{cases}
 [\partial_{ \zeta}y_{i}( \zeta),\partial_{ \omega}y_{j}( \omega)]=-2\delta_{ij}\partial_{ \omega}\delta( \zeta- \omega)-j\epsilon_{ijk}\partial_{ \omega}y_{k}\delta( \zeta- \omega)+{\color{black}O(j^{2})}\\
 [\partial_{ \zeta}y_{i}( \zeta),\partial_{\bar{ \omega}}y_{j}(\bar{ \omega})]=0+{\color{black}O(j^{2})}\\
 [\partial_{\bar{ \zeta}}y_{i}(\bar{ \zeta}),\partial_{\bar{ \omega}}y_{j}(\bar{ \omega})]=2\delta_{ij}\partial_{\bar{ \omega}}\delta(\bar{ \zeta}-\bar{ \omega})-j\epsilon_{ijk}\partial_{\bar{ \omega}}y_{k}\delta(\bar{ \zeta}-\bar{ \omega})+{\color{black}O(j^{2})}
\end{cases}
\end{split}
\begin{split}
 \begin{cases}
T(\ \zeta):=-\frac{1}{4}\partial_{\zeta}y_{i}(\zeta)\partial_{\zeta}y_{i}(\zeta)\\
T(\bar{\zeta}):=\frac{1}{4}\partial_{\bar{\zeta}}y_{i}(\bar{\zeta})\partial_{\bar{\zeta}}y_{i}(\bar{\zeta})
\end{cases}
\end{split}
\label{9}
\end{equation}}

where $T(\zeta)$ and $T(\bar{\zeta})$ are two copies of the Virasoro algebra up to first order in $j$ with central charge dim$\mathbb{T}^{3}=3$.

The relations in (\ref{9}) were found, from a different point of view,  in \cite{blumenhagen2011non}. We have arrived at this expression using the Hamiltonian interpretation of vertex algebras. The equations in (\ref{6}) are implicit in their work, where they used the correlators language. Finally, for the general case $k\neq 0$ and $j\neq 0$ the beta equations are not satisfied at any order therefore a similar procedure is not possible in the general case. 


 \section{Proof of the singularity}\label{proof}
  
In this section we prove that the logarithmic singularities in the theorem \ref{principal} imply the algebraic relations of the affine Kac Moody vertex algebra $V^{1}(\mathfrak{h}_{k,j})$ given in \eqref{1}. First, we prove this result for the particular case, example 3.2, $k=0$ and $j=1$.

\subsection{The \emph{double twisted torus} case $k=0$, $j=1$}

First, we differentiate the identity\footnote{Note that we have the identities $$\partial_{z}log(z-w)=-\partial_{w}log(z-w)=\delta(z,w),\quad\quad \partial_{w}\partial_{z}rl(z,w)=-\frac{1}{2}(\log z-\log w)\partial_{w}\delta(z,w)$$} $[y_{1}(z),y_{3}(w)]$ in the theorem (\ref{principal})
\begin{align*}
    &[\partial_{z}y_{i}(z),\partial_{w}y_{j}(w)]=\partial_{z}\partial_{w} ( \varepsilon_{ijk}w^{k}rl(z,w)+\frac{1}{2}\varepsilon_{ijk}(\hat{x}_{k}(z)-\hat{x}_{k}(w))\log (z-w)) \\
    &=\varepsilon_{ijk}w^{k}(-\frac{1}{2}(logz-logw)\partial_{w}\delta(z,w))+\frac{1}{2}(-\partial_{z}\hat{x}_{3}(z)-\partial_{w}\hat{x}_{3}(w))\delta(z,w)+\frac{1}{2}(\hat{x}_{3}(z)-\hat{x}_{3}(w))\partial_{w}\delta(z,w).
\end{align*}

Then we compute the same expression using the equations (\ref{formas})
\begin{align*}
&[\partial_{z}y_{i}(z),\partial_{w}y_{j}(w)]=[\beta_{1}(z)+\frac{1}{2}x_{2}\partial_{z}x_{3}(z)-\frac{1}{2}x_{3}\partial_{z}x_{2}(z), \beta_{2}(z)+\frac{1}{2}x_{3}\partial_{z}x_{1}(z)-\frac{1}{2}x_{1}\partial_{z}x_{3}(z) ]\\
&=\alpha_{3}(w)\delta(z,w)+\frac{1}{2}x_{3}(w)\partial_{w}\delta(z,w)-\frac{1}{2}\delta(z,w)\partial_{w}x_{3}(w)-\frac{1}{2}\delta(z,w)\partial_{z}x_{3}(z)-\frac{1}{2}x_{3}(z)\partial_{w}\delta(z,w),
\end{align*}
where we used that $[x_{2}(z),\beta_{2}(w)]=[x_{2}(z),\partial_{w}y_{2}(w)]=-\delta(z,w)$ and in the same way $[\beta_{1}(z),x_{1}(w)]=\delta(z,w)$.

The fact that these two results are the same follows from the next theorem

\begin{teo}\label{kac}
Let $a(z)$ be a formal distribution and let N be a non-negative integer. Then one has the following equality of formal distributions in z and w:
$$\partial_{w}^{N}\delta(z-w)a(z)=\partial_{w}^{N}\delta(z-w)\sum_{j=0}^{N}\partial^{j}a(w)(z-w)^{j}$$

\end{teo}
See \cite{kacvertex} for a proof of this theorem.

\subsection{General case}

First, we compute the expression differentiating the identity\footnote{Note that we have the identities 
\begin{align*}\partial_{z}rl(z,w)&=\frac{1}{2}\frac{1}{z}log(z-w)-\frac{1}{2}(logz-logw)\delta(z,w)\end{align*}
\begin{align*}
    &\partial_{z}\partial_{w}t(z,w)=-\frac{1}{6}(logz-logw)(\frac{1}{z}\delta(z,w))+\frac{1}{6}\{\frac{3}{wz}log(z-w)\}-\frac{1}{6}(logz^{2}-3logzlogw+logw^2)\partial_{w}\delta(z,w)
\end{align*}

\begin{align*}
    &\partial_{z}\partial_{w}t(z,w)=\partial_{z}\partial_{w}(\sum_{k}(\frac{z}{w})^{k}\frac{1}{k^{3}}-\sum_{k}(\frac{w}{z})^{k}\frac{1}{k^{3}})(2)+\partial_{z}\partial_{w}((logz-logw)(Li_{2}(\frac{z}{w}))+Li_{2}(\frac{w}{z})))\\
    &+\partial_{z}\partial_{w}\{\frac{1}{6}(logz-logw)^{3}\}-\partial_{z}\partial_{w}\{\frac{1}{6}(logz^{2}-3logzlogw+logw^2)log(z-w)\}\\
    &=2(\frac{1}{zw}(log(z-w)-logz+logw))-2(\frac{1}{zw}(log(z-w)-logz+logw))+(logz-logw)(-\frac{1}{z}\delta(z,w)+\frac{1}{zw})\\
       &+\{(logz-logw)(\frac{-1}{zw})\}+\frac{1}{6}\{\frac{3}{wz}log(z-w)-\frac{1}{6}(-3\frac{logz}{w}+2\frac{logw}{w}-2\frac{logz}{z}+3\frac{logw}{z})\delta(z,w)\}\\
    &-\frac{1}{6}(logz^{2}-3logzlogw+logw^2)\partial_{w}\delta(z,w)\\
    &=-\frac{1}{6}(logz-logw)(\frac{1}{z}\delta(z,w))+\frac{1}{6}\{\frac{3}{wz}log(z-w)\}-\frac{1}{6}(logz^{2}-3logzlogw+logw^2)\partial_{w}\delta(z,w)
\end{align*}



} in the theorem \ref{principal}
\begin{align*}
&\partial_{z}\partial_{w}[y_{1}(z),y_{3}(w)]=\\
      &=-\frac{j}{2}w_{2}(-(\log z-\log w)\partial_{w}\delta(z,w))+\frac{j}{2}(\partial_{z}\hat{x}_{2}(z)+\partial_{w}\hat{x}_{2}(w))\delta(z,w)+\frac{j}{2}(\hat{x}_{2}(w)-\hat{x}_{2}(z))\partial_{w}\delta(z,w)\\
      &-\frac{k}{2}p_{2}(-(\log z-\log w)\partial_{w}\delta(z,w))+\frac{k}{2}(\partial_{z}\hat{y}_{2}(z)+\partial_{w}\hat{y}_{2}(w))\delta(z,w)+\frac{k}{2}(\hat{y}_{2}(w)-\hat{y}_{2}(z))\partial_{w}\delta(z,w)\\
      &+\frac{kj}{6}({x}_{3}(z){x}_{1}(z)+{x}_{3}(w){x}_{1}(w))\partial_{w}\delta(z,w)+\frac{kj}{6}(\partial_{w}({x}_{3}(w){x}_{1}(w))-\partial_{z}({x}_{3}(z){x}_{1}(z)))\delta(z,w)\\
      &-\frac{kj}{2}\left(\partial_{z}\hat{x}_{3}(z)\partial_{w}\hat{x}_{1}(w)\log (z-w)+\hat{x}_{3}(z)\partial_{w}\hat{x}_{1}(w)\delta(z,w)-\partial_{z}\hat{x}_{3}(z)\hat{x}_{1}(w)\delta(z,w)+\hat{x}_{3}(z)\hat{x}_{1}(w)\partial_{w}\delta(z,w)\right)\\
      &-\frac{kj}{2}(-w_{3}\partial_{z}(\hat{x}_{1}(z)\log z)+w_{1}\partial_{w}(\hat{x}_{3}(w)\log w))\delta(z,w)-\frac{kj}{2}(w_{3}\hat{x}_{1}(z)\log z+w_{1}\hat{x}_{3}(w)\log w)\partial_{w}\delta(z,w) \\
      &+kj(w_{1}\partial_{z}\hat{x}_{3}(z))(-\frac{1}{2w}\log (z-w)+\frac{1}{2}(\log z-\log w)\delta(z,w))\\
      &-kj(w_{3}\partial_{w}\hat{x}_{1}(w))(\frac{1}{2z}\log (z-w)-\frac{1}{2}(\log z-\log w)\delta(z,w))\\
      &+kj(w_{1}\hat{x}_{3}(z)-w_{3}\hat{x}_{1}(w))(-\frac{1}{2}(\log z-\log w)\partial_{w}\delta(z,w))\\
      &-kjw_{3}w_{1}(-\frac{1}{2}(\log z-\log w)\frac{1}{z}\delta(z,w)+\frac{1}{6}\frac{3}{wz}\log (z-w)+\frac{1}{2}\log z\log w\partial_{w}\delta(z,w))
\end{align*}

Then we compute the same expression using the equations in (\ref{formas})
\begin{align*}
 &[\partial_{z}y_{1}(z),\partial_{w}y_{3}(w)]=\\
 &-(j\alpha_{2}(w)+k\beta_{2}(w))\delta(z,w)+\frac{j}{2}(\partial_{w}x_{2}(w)+\partial_{z}x_{2}(z))\delta(z,w)+\frac{k}{2}(x_{2}(z)-x_{2}(w)\partial_{w}\delta(z,w))\\
 &\frac{j}{2}(\partial_{w}x_{2}(w)\delta(z,w)-x_{2}(w)\partial_{w}\delta(z,w))+\frac{k}{2}(\partial_{w}y_{2}(w)\delta(z,w)-y_{2}(w)\partial_{w}\delta(z,w))\\
 &\frac{kj}{2}((x_{1}(w)\partial_{w}x_{3}(w)-x_{3}(w)\partial_{w}x_{1}(w))\delta(z,w)+x_{1}(w)x_{3}(w)\partial_{w}\delta(z,w))\\
 &\frac{j}{2}(\partial_{z}x_{2}(z)\delta(z,w)+x_{2}(z)\partial_{w}\delta(z,w))+\frac{k}{2}(\partial_{z}y_{2}(z)\delta(z,w)+y_{2}(z)\partial_{w}\delta(z,w))\\
 &\frac{kj}{2}((x_{1}(z)\partial_{z}x_{3}(z)-x_{3}(z)\partial_{z}x_{1}(z))\delta(z,w)+x_{1}(z)x_{3}(z)\partial_{w}\delta(z,w))\\
  &\frac{kj}{3}((2x_{1}(w)\partial x_{3}(w)-x_{3}(w)\partial_{w}x_{1}(w))\delta(z,w)-x_{3}(w)x_{1}(w)\partial_{w}\delta(z,w))\\
  &\frac{kj}{3}((-2x_{3}(z)\partial x_{1}(z)+x_{1}(z)\partial x_{3}(z))\delta(z,w)-x_{3}(z)x_{1}(z)\partial_{w}\delta(z,w))\\
 &\frac{kj}{2}(-\partial x_{3}(z)x_{1}(w)\delta(z,w)-x_{3}(z)x_{1}(w)\partial_{w}\delta(z,w)-\partial x_{3}(z)\partial x_{1}(w)\log (z-w)+x_{3}(z)\partial_{w}x_{1}(w)\delta(z,w)),
  \end{align*}
where we used $[\beta_{1}(z),x_{2}(w)]=\frac{k}{2}x_{3}(w)\delta(z,w)$, $[\beta_{1}(z),y_{2}(w)]=\frac{j}{2}x_{3}(w)\delta(z,w)$ and  $[x_{2}(z),\beta_{3}(w)]=\frac{k}{2}x_{1}(z)\delta(z,w)$ and $[y_{2}(z),\beta_{3}(w)]=\frac{j}{2}x_{1}(z)\delta(z,w)$.

The fact that these two results are the same follows from theorem \ref{kac}.


\begin{thebibliography}{10}
\bibitem{aldi2012dilogarithms}  {Aldi, Marco and Heluani, Reimundo},
  {\em Dilogarithms, OPE, and Twisted T-Duality},
  {Int. Math. Res. Notices},  vol {6}, {1528--1575}, {2012}. {{\tt https://arxiv.org/abs/1105.4280}}
 

\bibitem{blumenhagen2011non}{ Blumenhagen, Ralph and Deser, Andreas and L{\"u}st, Dieter and Plauschinn, Erik and others},
  {\em Non-geometric fluxes, asymmetric strings and nonassociative geometry},
  { J. Phys. A: Mathe. Theor.}, vol {44}, {2011}. {{\tt https://arxiv.org/abs/1106.0316}}

\bibitem{kapustin2003vertex}  {Kapustin, Anton and Orlov, Dmitri},
  {\em Vertex algebras, mirror symmetry, and D-branes: the case of complex tori},
  { \emph Comm. Math. Phys},
  vol {233}, {2003}. {{\tt https://arxiv.org/abs/hep-th/0010293}}

\bibitem{alekseev2005current}  {Alekseev, Anton and Strobl, Thomas},
  {\em Current algebras and differential geometry},
  {\emph J. High Energy Phys},
  vol {2005},
  {035},
  {2005}. {{\tt https://arxiv.org/abs/hep-th/0410183}}

\bibitem{reid2009flux} {Reid-Edwards, RA},
  {\em Flux compactifications, twisted tori and doubled geometry},
  {\emph J. High Energy Phys},
  vol {2009},
  {085},
  {2009}.{{\tt https://arxiv.org/abs/0904.0380}}

\bibitem{bouwknegt2004t}{Bouwknegt, Peter and Evslin, Jarah and Mathai, Varghese},
  {\em T-duality: topology change from H-flux},
  {\emph Comm. Math. Phys},
 vol {249},
  {2},
    {383--415},
    {2004}. {{\tt https://arxiv.org/abs/hep-th/0306062}}

\bibitem{cavalcanti2011generalized}     {Cavalcanti, Gil R and Gualtieri, Marco},
  {\em Generalized complex geometry and T-duality},
  {\tt arXiv preprint arXiv:1106.1747},
    {2011}.


\bibitem{hekmati2012t} {Hekmati, Pedram and Mathai, Varghese},
  {\em T-duality of current algebras and their quantization}, 
  {Contemp. Math} {584}, {17-38}, {2012}.{{\tt https://arxiv.org/abs/1203.1709}}
    
    \bibitem{vsevera2015poisson}    {{\v{S}}evera, Pavol},
  {\em Poisson-Lie T-duality and Courant algebroids},
  {Lett. Math. Phys.}, {12}, {105}, {1689--1701},{2015}{\tt arXiv preprint arXiv:1502.04517}.


\bibitem{kacvertex}  {Kac, V.G.},
  {\em Vertex Algebras for Beginners},
 {Univ. Lecture Ser},
 

\bibitem{frenkel2004vertex}{Frenkel, E. and Ben-Zvi, D.},
  {\em Vertex Algebras and Algebraic Curves: Second Edition: },
  {Math. Surveys Monogr},
    {2004}.

\bibitem{lepowsky2004introduction}  {Lepowsky, J. and Li, H.},
  {\em Introduction to Vertex Operator Algebras and Their Representations},
     {Birkh{\"a}user Boston}, {2004}.



\bibitem{dorfman1987dirac}     {Dorfman, Irene Ya},
  {\em Dirac structures of integrable evolution equations},
  {\emph Modern Phys. Lett. A},
 vol {125},
  {5},
    {240--246},
    {1987}.

\bibitem{de2006finite}  {De Sole, Alberto and Kac, Victor G},
  {\em Finite vs affine W-algebras},
  {\emph Jpn. J. Math},
 vol {1},
  {1},
    {137--261},
    {2006}.{{\tt https://arxiv.org/abs/math-ph/0511055}}


\bibitem{hull2009double} {Hull, Chris and Zwiebach, Barton},
  {\em Double field theory},
  {\emph J. High Energy Phys.},
 vol {2009},
    {099},
    {2009}.{{\tt https://arxiv.org/abs/0904.4664}}




\bibitem{ekstrand2009non}{Ekstrand, Joel and Heluani, Reimundo and K{\"a}ll{\'e}n, Johan and Zabzine, Maxim and others},
  {\em Non-linear sigma models via the chiral de Rham complex},
  {\emph Adv. Theor. Math. Phys},
 vol {13},
  {4},
    {1221--1254},
    {2009}.{{\tt https://arxiv.org/abs/0905.4447}}


\bibitem{green1988superstring} {Green, M.B. and Green, M.B. and Schwarz, J.H. and Witten, E.},
  {\em Superstring Theory: Volume 1, Introduction},
 {Cambridge Monogr. Math. Phys},
      {1988}.

\bibitem{malikov1999chiral} {Malikov, Fyodor and Schechtman, Vadim and Vaintrob, Arkady},
  {\em Chiral de Rham complex},
  {\emph Comm. Math. Phys.},
 vol {204},
  {2},
    {439--473},
    {1999}. {{\tt https://arxiv.org/abs/math/9803041}}









\end{thebibliography}
    \end{document}